\documentclass[10pt,twoside,a4paper]{amsart}
\usepackage{amsmath}
\usepackage{amsfonts}
\usepackage{amssymb}
\usepackage{amsthm}
\usepackage{newlfont}
\usepackage{graphicx}
\usepackage{amscd}

\theoremstyle{plain}
\newtheorem{Th}{Theorem}%[section]

\newtheorem{Lem}[Th]{Lemma}
\newtheorem{Prop}[Th]{Proposition}
\theoremstyle{definition}
\theoremstyle{remark}
\newtheorem*{Rem}{Remark}%[section]
\numberwithin{equation}{section}
\newcommand{\RR}{{\mathbb R}}

\newcommand{\CC}{{\mathbb C}}
\newcommand{\ZZ}{{\mathbb Z}}

\newcommand{\bx}{{\boldsymbol x}}
\newcommand{\bX}{{\boldsymbol X}}

\begin{document}

\title
{On $\tau$-function of the quadrilateral lattice}

\author{Adam Doliwa}

\address{Adam Doliwa, Wydzia{\l} Matematyki i Informatyki,
Uniwersytet Warmi\'{n}sko-Mazurski w Olsztynie,
ul.~\.{Z}o{\l}nierska~14, 10-561 Olsztyn, Poland}

\email{doliwa@matman.uwm.edu.pl}

%\today 

\begin{abstract}

We investigate the $\tau$-function of the
quadrilateral lattice using the nonlocal
$\bar\partial$-dressing method, and we show that it can be identified with the 
Fredholm determinant of the
integral equation which naturally appears within that approach. \\

{\it 2008 PACS:} 02.30.Ik, 02.30.Rz, 05.45.Yv 

\end{abstract}
\maketitle

\section{Introduction}
Lattices of planar quadrilaterals were introduced into integrability theory in
\cite{DCN,MQL} as discrete analogs \cite{Sauer2,Sauer} of the conjugate nets
\cite{DarbouxOS}. The basic nonlinear system describing such lattices, 
the discrete Darboux equations 
\eqref{eq:MQL-A} or \eqref{eq:MQL-Q}, appeared however first in \cite{BoKo} as
the generic discrete systems integrable by the nonlocal $\bar\partial$-dressing
method \cite{AYF,ZakMa,BogdanovManakov,Konopelchenko}. The classical Darboux
equations \cite{DarbouxOS}, called nowadays also the $N$-wave system 
\cite{ZakMa,KvLJMP}, form the basic system of equations of the multicomponent
Kadomtsev-Petviashvili (KP) hierarchy \cite{DKJM,KvLJMP}. The Darboux equations were
rediscovered \cite{ZakMa} in the
generalized matrix form, as the basic set of equations solvable by the 
$\bar\partial$-dressing method. For example, the 
KP equation \cite{DKJM} was shown in 
\cite{BogdanovManakov} to be a limiting case of the Darboux system. 
In fact the whole KP hierarchy
can be written, in the so called Miwa coordinates, as an infinite system of 
the Darboux equations \cite{BK-mKP}.

In \cite{tau-cn} it was shown that the $\tau$-function of
conjugate nets, which is a potential apparently known already 
to Darboux, and which can be viewed as the $\tau$-function of the whole 
multicomponent KP
hierarchy upon identification of the higher times of the hierarchy with
isoconjugate deformations of the nets \cite{DMMMS}, can be
identified 
with the Fredholm determinant of the integral equation inverting the 
nonlocal $\bar\partial$-problem as applied to Darboux equations. 
The $\tau$-functions play the central role
\cite{Sato,DKJM,SegalWilson,GrinevichOrlov,Kharchev,
Palmer,ASvM,BorodinDeift,Moerbeke} in
establishing the connections between integrable systems and quantum field
theory, statistical mechanics or the theory of random matrices. They
are often represented as determinants of infinite matrices or can be identified 
with the Fredholm determinant of the integral Gel'fand--Levitan--Marchenko
equation used to solve the model under consideration. Within the context 
of the  Zakharov and Shabat dressing
method \cite{ZMNP} the $\tau$-function of the KP hierarchy was
interpreted as the Fredholm determinant in \cite{PoppeSattinger}.

The
$\bar\partial$-dressing and related methods have been applied to conjugate nets,
quadrilateral lattices and their reductions and transformations in a number of
papers, including apart from the mentioned above also
\cite{ZakMa2,DMS,TQL,DMMMS,DS-sym,MMAM-cn}. The main result of the present paper,
missed somehow in earlier studies, is the interpretation of the $\tau$-function
of the quadrilateral lattice as the Fredholm determinant of the integral 
equation inverting the nonlocal $\bar\partial$-problem as studied in
\cite{BoKo}. In fact, as the discrete Darboux equations can be obtained from the
multicomponent KP hierarchy via the Miwa transformation \cite{DMMMS}, the result
is not a surprise. However, because integrable discrete systems are often considered
as more fundamental then their differential counterpart, and because the  
nonlocal $\bar\partial$-approach contains,
as particular reductions, earlier versions of the inverse
spectral (scattering) method, it was desirable to
have a direct proof. 

The paper is constructed as follows. In
Section~\ref{sec:dbar-QL} we collect the basic elements of the 
$\bar\partial$-dressing method 
and we recall the way of solving the discrete
Darboux equations within this method.   
Section~\ref{sec:Fdet-tau} is devoted to presentation of the
Fredholm determinant interpretation of the
$\tau$-function of the quadrilateral lattice. In the remaining part of the
Indroduction we present basic elements of the quadrilateral lattice theory.
 
In affine representation the quadrilateral lattice is a mapping
$\bx:\ZZ^N\to\RR^M$, and the planarity of its elementary quadrilaterals can be
formulated in terms of the system of the discrete Laplace equations
\begin{equation}  \label{eq:Laplace}
\Delta_i\Delta_j\bx=(T_{i} A_{ij})\Delta_i\bx+
(T_j A_{ji})\Delta_j\bx,\quad i\not= j, \quad  i,j=1 ,\dots, N ,
\end{equation}
where $T_i$ denotes the shift operator in $n_i$, and $\Delta_i = T_i -1$
is the corresponding partial difference. For $N >2$ the coefficients 
$A_{ij}$ of the Laplace
equation \eqref{eq:Laplace} satisfy the discrete Darboux  
equations
\begin{equation} \label{eq:MQL-A}
\Delta_k A_{ij} =
 (T_jA_{jk})A_{ij} +(T_k A_{kj})A_{ik} - (T_kA_{ij})A_{ik},
\;\; i\neq j\neq k\neq i ,
\end{equation}
which is the integrable discrete analog of the Darboux equations describing
multidimensional conjugate nets~\cite{DarbouxOS}.
Equations \eqref{eq:MQL-A} imply existence of the potentials
$H_i$~\cite{BoKo,MQL}
\begin{equation}   \label{def:A-H}
A_{ij}= \frac{\Delta_j H_i}{H_i} , \qquad i\ne j ,
\end{equation}
which, in analogy to the continuous case, are called the 
Lam\'e coefficients.

Introduce the suitably scaled tangent 
vectors $\bX_i$, $i=1,...,N$,
\begin{equation}  \label{def:HX}
\Delta_i\bx = (T_iH_i) \bX_i,
\end{equation}
then
\begin{equation} \label{eq:lin-X}
\Delta_j\bX_i = (T_j Q_{ij})\bX_j,    \qquad i\ne j \; ,
\end{equation}
and the compatibility condition for the system \eqref{eq:lin-X}
gives the following new form of the discrete Darboux equations
\begin{equation} \label{eq:MQL-Q}
\Delta_kQ_{ij} = (T_kQ_{ik})Q_{kj}, \qquad i\neq j\neq k\neq i.
\end{equation}
The 
the Lam\'e coefficients $H_i$ solve the linear equations
\begin{equation} \label{eq:lin-H}
\Delta_iH_j = (T_iH_i) Q_{ij}, \qquad i\ne j \; ,
\end{equation}
 whose compatibility gives equations
\eqref{eq:MQL-Q} again. 

Define \cite{KoSchief2,DS-sym} the  potentials $\rho_i$ as solutions of compatible equations
\begin{equation} \label{eq:rho-constr}
\frac{T_j\rho_i}{\rho_i} = 1 - (T_iQ_{ji})(T_jQ_{ij}), \qquad  i\ne j.
\end{equation}
The right hand side of equation \eqref{eq:rho-constr} is symmetric with
respect to the interchange of $i$ and $j$, which implies the
existence of a potential $\tau$, such that
\begin{equation} \label{eq:rho-tau}
\rho_i = \frac{T_i\tau}{\tau} ,
\end{equation}
which is called
the $\tau$-{\it function} of the quadrilateral
lattice~\cite{DS-sym}.

\section{The $\bar\partial$-dressing method and the discrete Darboux equations}
\label{sec:dbar-QL}
In this Section we recall \cite{ZakMa,BogdanovManakov,Konopelchenko}
the basic idea of the nonlocal $\bar\partial$-dressing method in application to
the quadrilateral lattice and the discrete Darboux equations \cite{BoKo}.

Consider the following
integro-differential equation in the complex plane ${\mathbb C}$
\begin{equation} \label{eq:db-nonl}
\bar\partial \chi(\lambda) = \bar\partial \eta(\lambda) + 
\int_{\mathbb C}  R(\lambda,\lambda')\chi(\lambda') \: 
d\lambda'\wedge d\bar\lambda',
\end{equation}
where $R(\lambda,\lambda')$ is a given $\bar\partial$ datum, which decreases
quickly enough at $\infty$ in $\lambda$ and $\lambda^\prime$, and the 
function $\eta(\lambda)$, the normalization of the unknown $\chi(\lambda)$, is a
given rational function, which describes the polar behavior of $\chi(\lambda)$ in $\CC$
and its behavior at $\infty$: 
\begin{equation*}
\chi(\lambda) - \eta(\lambda) \to 0, \qquad 
\text{for} \quad |\lambda|\to\infty.
\end{equation*}
We remark that the dependence of $\chi(\lambda)$ and $R(\lambda,\lambda')$ on
$\bar\lambda$ and $\bar\lambda^\prime$ will be systematically omitted, for
notational convenience.

Due to the generalized Cauchy formula the nonlocal $\bar\partial$ 
problem~\eqref{eq:db-nonl} 
is equivalent to the following Fredholm integral equation of the 
second kind  
\begin{equation} \label{eq:Fredh-db}
\chi(\lambda) = \eta(\lambda) - \int_{\mathbb C} K(\lambda, \lambda') 
\chi(\lambda') d\lambda'\wedge d\bar\lambda',
\end{equation}
with the kernel
\begin{equation} \label{eq:Fredh-ker-db}
K(\lambda, \lambda') = \frac{1}{2\pi i} \int_{\mathbb C} 
\frac{R(\lambda'', \lambda') }{\lambda'' - \lambda}
d\lambda''\wedge d\bar\lambda''  .  
\end{equation}
Recall (see, for example \cite{Smithies}) that 
the Fredholm determinant $D_F$ is defined by the series
\begin{equation}  \label{eq:Fr-det}
D_F = 1 +
\sum_{m=1}^\infty \frac{1}{m!} \int_{\CC^m}
K \begin{pmatrix} \zeta_1 & \zeta_2 & \dots  & \zeta_m \\
\zeta_1 & \zeta_2 & \dots  & \zeta_m
\end{pmatrix} d\zeta_1 \wedge d\bar\zeta_1 \dots d\zeta_m \wedge d\bar\zeta_m \; ,
\end{equation}
where
\begin{equation*}  \label{eq:Fr-KK}
K \begin{pmatrix} \zeta_1 & \zeta_2 & \dots  & \zeta_m \\ 
\mu_1 & \mu_2 & \dots  & \mu_m
\end{pmatrix} = \det \begin{pmatrix} 
K(\zeta_i,\mu_j)
\end{pmatrix}_{1\leq i,j \leq m} .
\end{equation*}
For a nonvanishing Fredholm determinant
the solution of 
\eqref{eq:Fredh-db} can be written in the form
\begin{equation} \label{eq:F-sol}
\chi(\lambda) = \eta(\lambda) - \int_\CC \frac{D_F(\lambda,\lambda^\prime)}{D_F}
\eta(\lambda^\prime) d\lambda'\wedge d\bar\lambda',
\end{equation}
where the Fredholm minor is defined by the series
\begin{equation}  \label{eq:Fr-min}
D_F(\lambda,\lambda^\prime) = \sum_{m=0}^\infty \frac{1}{m!} 
\int_{\CC^m} K \begin{pmatrix} 
\lambda & \zeta_1 &  \dots  & \zeta_m \\ 
\lambda^\prime & \zeta_1 &  \dots  & \zeta_m
\end{pmatrix} d\zeta_1 \wedge d\bar\zeta_1 \dots d\zeta_m \wedge d\bar\zeta_m  \; .
\end{equation}

Let $\lambda_i^\pm\in\CC$, $i=1,\dots,N$ be pairs of distinct points of the
complex plane. To get the $\bar\partial$-dressing method 
of construction of solutions to the discrete Darboux equations 
one introduces \cite{BoKo} the following 
dependence of the kernel $R$ on the variables $n=(n_1,\dots ,n_N)\in\ZZ^N$
\begin{equation} \label{eq:evol-R-dis-1}
T_i R(\lambda,\lambda^\prime; n) =  L_i(\lambda)^{-1} 
R(\lambda,\lambda^\prime) L_i(\lambda^\prime),
\qquad L_i(\lambda) = \frac{\lambda -\lambda^-_i}{\lambda - \lambda^+_i},
\end{equation}
or equivalently 
\begin{equation} \label{eq:evol-R-dis}
R(\lambda,\lambda^\prime; n) = 
G(\lambda; n)^{-1} 
R_0(\lambda,\lambda^\prime) G(\lambda^\prime; n),
\qquad
G(\lambda;n)= \prod_{i=1}^N L_i(\lambda)^{n_i} ,
\end{equation}
where $R_0(\lambda,\lambda^\prime)$ is independent of $n$. We assume that $R_0$
decreases at $\lambda_i^\pm$ and in poles of the normalization function $\eta$
fast enough such that $\chi - \eta$ is regular in these points
\cite{BogdanovManakov,BoKo}.
\begin{Rem}
In the paper we always assume that the kernel $R$ in the nonlocal 
$\bar\partial$
problem is such that the Fredholm equation \eqref{eq:Fredh-db}
is uniquely solvable. Then, by the Fredholm alternative, the homogenous 
equation with $\eta=0$ has only the trivial solution.
\end{Rem}

Directly one can verify the following result which gives evolution of the
kernel $K$ of the Fredholm equation \eqref{eq:Fredh-db} implied by evolution of
the $\bar\partial$ datum.
\begin{Lem}
The evolution \eqref{eq:evol-R-dis-1} of the kernel $R$ implies that the
kernel $K$ of the integral equation \eqref{eq:Fredh-db}
is subject to the equation 
\begin{equation} \label{eq:evol-K-dis} 
T_i K(\lambda,\lambda^\prime; n)  = 
L_i(\lambda)^{-1} \left[ K(\lambda,\lambda^\prime; n) + \left( L_i(\lambda) -1
\right) K(\lambda^-_i,\lambda^\prime;n) \right] L_i(\lambda^\prime), 
\end{equation}
moreover
\begin{equation} \label{eq:Ti-K-l+}
T_iK(\lambda^+_i,\lambda^\prime) = 
K(\lambda^-_i,\lambda^\prime;n) L_i(\lambda^\prime).
\end{equation}
\end{Lem}
This leads to the following crucial, for our purposes, result.
\begin{Lem} 
When $\chi(\lambda;n)$ is the unique, by assumption, solution of 
the $\bar\partial$ problem \eqref{eq:db-nonl} with the kernel $R$
evolving according to \eqref{eq:evol-R-dis-1}, and  
with normalization $\eta(\lambda;n)$, then the function
$L_i(\lambda) T_i \chi (\lambda;n) $
is the solution of the same $\bar\partial$ problem
but with  the new normalization
\begin{equation} 
\eta^{(i)}(\lambda;n) =
L_i(\lambda) T_i \eta(\lambda;n) + (L_i(\lambda)-1) 
\lim_{\lambda\to\lambda^+_i}\left(T_i \chi(\lambda;n) - T_i \eta(\lambda;n) 
\right).
\end{equation}
\end{Lem}
\begin{Rem}
We allow for the normalization to have poles at the distinguished points
$\lambda^\pm_i$ of the construction. 
\end{Rem}
\begin{proof}
Application of equations \eqref{eq:evol-K-dis} and \eqref{eq:Ti-K-l+} to
the shifted formula \eqref{eq:Fredh-db} leads to
%\begin{equation*}
\begin{align*}
L_i(\lambda) T_i & \chi (\lambda;n) =  L_i(\lambda) T_i \eta(\lambda;n) 
- \int_{\mathbb C} K(\lambda, \lambda';n) 
L_i(\lambda')T_i \chi(\lambda';n) d\lambda'\wedge d\bar\lambda' \\&
- (L_i(\lambda)-1) \int_{\mathbb C} T_i K(\lambda^+_i, \lambda';n) 
L_i(\lambda')T_i \chi(\lambda';n) d\lambda'\wedge d\bar\lambda'.
\end{align*}
%\end{equation*}
To obtain the statement notice that
\begin{equation*}
\lim_{\lambda\to\lambda^+_i}\left(T_i \chi(\lambda;n) - T_i
\eta(\lambda;n)\right)
= -  \int_{\mathbb C} T_i K(\lambda^+_i, \lambda';n) 
L_i(\lambda')T_i \chi(\lambda';n) d\lambda'\wedge d\bar\lambda'.
\end{equation*}
\end{proof}

The following result allows to give the $\bar\partial$-method of 
construction of solutions of the discrete Darboux equations
\begin{Prop}
Let $\chi_i(\lambda;n)$, $i=1,\dots,N$ be solution of
the $\bar\partial$ problem \eqref{eq:db-nonl} with the normalization
\begin{equation} \label{eq:n-chi-i}
\eta_i(\lambda) =L_i(\lambda)-1 = \frac{\lambda_i^+ - \lambda_i^-}{\lambda -
\lambda_i^+},
\end{equation}
denote
\begin{equation}
G_i(\lambda;n)= \prod_{j=1, j\ne i}^N L_j(\lambda)^{n_j} ,
\end{equation}
then the functions
\begin{align}
\psi_i(\lambda;n) & = \chi_i(\lambda;n) G(\lambda;n) G_i(\lambda_i^+;n)^{-1}, \\
Q_{ij}(n) & = \chi_i(\lambda_j^+;n) 
\frac{G_j(\lambda_j^+ ;n)}{G_i(\lambda_i^+;n)} L_j(\lambda_i^+).
\label{eq:def-Qij}
\end{align}
satisfy equations
\begin{align} \label{eq:psi-i-lin}
\Delta_j\psi_i(\lambda;n) & = T_j Q_{ij}(n) \psi_j(\lambda;n) , \qquad j\ne i,\\
\label{eq:Q-ij-nlin}
\Delta_j Q_{ik}(n) & = T_jQ_{ij}(n) Q_{jk}(n), \qquad i\ne j \ne k \ne i.
\end{align}
\end{Prop}
\begin{proof}
The combination $L_j(\lambda) T_j \chi_i (\lambda;n) L_j(\lambda_i^+)^{-1} -
\chi_i (\lambda;n)$, $j\ne i$ satisfies the integral equation \eqref{eq:Fredh-db} with the 
same normalization as $T_j \chi_i (\lambda_j^+) \chi_j(\lambda;n)$. Therefore, by
the Fredholm alternative,
\begin{equation} \label{eq:chi-i-lin}
L_j(\lambda) T_j \chi_i (\lambda;n) L_j(\lambda_i^+)^{-1} - \chi_i (\lambda;n)
= T_j \chi_i (\lambda_j^+) \chi_j(\lambda;n) ,
\qquad j\ne i,
\end{equation}
which leads to the linear
system \eqref{eq:psi-i-lin}. Its compatibility \eqref{eq:Q-ij-nlin}
can be also obtained by
evaluating equation \eqref{eq:chi-i-lin} in the points $\lambda_k^+$, $k\ne
i,j$.
\end{proof}
For completness we recall also the following result of \cite{BoKo}.
\begin{Th}
Let $\chi(\lambda;n)$ be solution of
the $\bar\partial$ problem \eqref{eq:db-nonl} with the canonical
normalization $\eta(\lambda) = 1$ then the function
\begin{equation*}
\psi(\lambda;n) = \chi(\lambda;n) G (\lambda;n)
\end{equation*}
satisfies the discrete Laplace system
\begin{equation*}
\Delta_i\Delta_j \psi(\lambda;n) = (T_iA_{ij})(n) \Delta_i \psi(\lambda;n) + 
(T_jA_{ji})(n) \Delta_j \psi(\lambda;n) , \qquad i\ne j,
\end{equation*}
with coefficients
\begin{equation*}  \label{eq:Aij-db}
A_{ij}(n) = L_j(\lambda_i^+)
\frac{T_j \chi(\lambda_i^+;n)}{\chi(\lambda_i^+;n)}-1, \qquad i\ne j,
\end{equation*}
while the corresponding Lam\'{e} coefficients are given by
\begin{equation*}
H_i(n) = \chi(\lambda_i^+;n)G_i(\lambda_i^+;n).
\end{equation*}
\end{Th}
\begin{Rem}
Various $n$-independent measures $\mathrm{d}\mu_a$ on $\CC$ give rise to
coordinates 
\begin{equation*}
x^a(n) = \int_\CC \psi(\lambda;n) \mathrm{d}\mu_a,
\end{equation*}
of quadrilateral lattices, having $H_i(n)$ as the Lam\'{e} coefficients,
and the functions
\begin{equation*}
X_i^a(n) = \int_\CC \psi_i(\lambda;n) \mathrm{d}\mu_a,
\end{equation*}
being coordinates of the normalized tangent vectors. To obtain real lattices,
the kernel $R_0$, the points $\lambda_i^\pm$, and the measures 
$\mathrm{d}\mu_a$ should satisfy certain additional conditions.
\end{Rem}

\section{The first potentials and the $\tau$-function}
\label{sec:Fdet-tau}

To give the meaning of the $\tau$-function within the $\bar\partial$-dressing
method we first present the meaning of the potentials $\rho_i$ defined
by equations \eqref{eq:rho-constr}. 
\begin{Prop} \label{prop:rho}
Within the $\bar\partial$-dressing method the potentials $\rho_i$ can be
identified with
\begin{equation}
\rho_i(n) = - \chi_i(\lambda_i^-;n) 
\frac{G_i(\lambda_i^-;n)}{G_i(\lambda_i^+;n)}.
\end{equation}
\end{Prop}
\begin{proof}
Evaluation of formula \eqref{eq:chi-i-lin} at $\lambda = \lambda_j^-$ gives
\begin{equation} \label{eq:chi-i-l-j}
\chi_i(\lambda_j^-;n) = - T_j \chi_i(\lambda_j^+;n) \chi_j(\lambda_j^-;n).
\end{equation}
Evaluation in turn of \eqref{eq:chi-i-lin} at $\lambda = \lambda_i^-$
gives 
\begin{equation*}
T_j\chi_i(\lambda_i^-;n) \frac{L_j(\lambda_i^-)}{L_j(\lambda_i^+)} -
\chi_i(\lambda_i^-;n) = T_j \chi_i(\lambda_j^+;n) \chi_j(\lambda_i^-;n),
\end{equation*}
which in view of \eqref{eq:chi-i-l-j} implies
\begin{equation*}
T_j\chi_i(\lambda_i^-;n) \frac{L_j(\lambda_i^-)}{L_j(\lambda_i^+)} = 
\chi_i(\lambda_i^-;n) 
\left(1 - T_j \chi_i(\lambda_j^+;n) T_i \chi_j(\lambda_i^+;n)  \right).
\end{equation*}
The last formula, the identification \eqref{eq:def-Qij}, and equation  
\eqref{eq:rho-constr} give the statement.
\end{proof}
Before proving that the $\tau$-function of the quadrilateral lattice equals
essentially to the Fredholm determinant of the integral equation inverting the
nonlocal $\bar\partial$ problem \eqref{eq:db-nonl} with the $\bar\partial$ datum
evolving according to the rule \eqref{eq:evol-R-dis-1}, we will need the
following technical result.
\begin{Lem} \label{lem:evol-KK}
The evolution \eqref{eq:evol-K-dis} of the kernel of the Fredholm equation
implies the following evolution of the determinants in the series defining the 
Fredholm determinant $D_F$ 
\begin{align*}
T_i & K\left( \begin{matrix} \zeta_1 & \zeta_2 & \dots  & \zeta_m \\
\zeta_1 & \zeta_2 & \dots  & \zeta_m \end{matrix} \biggr\rvert \; n \right)
= K\left( \begin{matrix} \zeta_1 & \zeta_2 & \dots  & \zeta_m \\
\zeta_1 & \zeta_2 & \dots  & \zeta_m \end{matrix} \biggr\rvert \; n \right) +\\
+&\sum_{j = 1}^m \left[ L_i(\zeta_j)-1\right] 
K\left( \begin{matrix} \lambda_i^- & \zeta_1 &  \dots  & \check{\zeta}_j
& \dots & \zeta_m \\
\zeta_j & \zeta_1 & \dots  & \check{\zeta}_j
& \dots & \zeta_m \end{matrix} \biggr\rvert \; n \right),
\end{align*}
where the symbol $\check{\zeta}_j$ means that $\zeta_j$ should be removed
from the sequence.
\end{Lem}
\begin{proof}
Extracting $L_i(\zeta_j)^{-1}$ from rows and $L_i(\zeta_k)$ from columns of the determinant 
\begin{equation*}
T_i  K\left( \begin{matrix} \zeta_1  & \dots  & \zeta_m \\
\zeta_1 &  \dots  & \zeta_m \end{matrix} \biggr\rvert \; n \right) =
\det \begin{pmatrix}
L_i(\zeta_j)^{-1} \left[ K(\zeta_j,\zeta_k; n) + \left( L_i(\zeta_j) -1
\right) K(\lambda^-_i,\zeta_k;n) \right] L_i(\zeta_k)
\end{pmatrix}_{1\leq j,k \leq m}, 
\end{equation*}
one arrives to the bordered determinant (an object often encountered in the
soliton theory \cite{Hirota-book})
\begin{align*}
\det \begin{pmatrix}
K(\zeta_j,\zeta_k; n) + \left( L_i(\zeta_j) -1
\right) K(\lambda^-_i,\zeta_k;n) \end{pmatrix}_{1\leq j,k \leq m} = \\
\det \begin{pmatrix}
K(\zeta_j,\zeta_k; n)  \end{pmatrix}_{1\leq j,k \leq m}
+ \sum_{j,k=1}^m \left( L_i(\zeta_j) -1 \right) \Delta_{jk}
K(\lambda^-_i,\zeta_k;n), 
\end{align*}
where by $(\Delta_{jk})_{1\leq j,k \leq m}$ we denote the cofactor matrix of
$(K(\zeta_j,\zeta_k; n))_{1\leq j,k \leq m}$. Notice that 
\begin{equation*}
\sum_{k=1}^m \Delta_{jk} K(\lambda^-_i,\zeta_k;n) = 
K\left( \begin{matrix} \zeta_1 &  \dots  & \lambda_i^- & 
& \dots & \zeta_m \\
 \zeta_1 & \dots  & \zeta_j &
& \dots & \zeta_m \end{matrix} \biggr\rvert \; n \right),
\end{equation*}
and, finally, application of an even number of transpositions to the above
determinant concludes the proof.
\end{proof}
From Lemma~\ref{lem:evol-KK} we
immediately obtain the evolution rule of the Fredholm determinant.
\begin{Prop} \label{prop:DF-evol}
The evolution \eqref{eq:evol-K-dis} of the kernel of the Fredholm equation
implies the following evolution of the 
Fredholm determinant $D_F$ 
\begin{equation}
T_i D_F(n) = D_F(n) + \int_\CC D_F(\lambda_i^-,\lambda;n)\left[
L_i(\lambda)-1\right] d\lambda \wedge d\bar\lambda .
\end{equation}
\end{Prop}
We are ready to state the main result of the paper.
\begin{Th}
Within the $\bar\partial$-dressing method with the $\bar\partial$ datum
evolving according to the rule \eqref{eq:evol-R-dis-1} 
the $\tau$-function of the quadrilateral lattice can be identified, up to
standard factor, 
with the Fredholm determinant as follows
\begin{equation}
\tau(n) = D_F(n) \prod_{i<j} A_{ij}^{n_i n_j}, \qquad
\text{where} \quad A_{ij} = \frac{L_i(\lambda_j^-)}{L_i(\lambda_j^+)} = A_{ji}.
\end{equation}
\end{Th}
\begin{proof}
Proposition \ref{prop:DF-evol}, formula \eqref{eq:F-sol}, 
and Proposition~\ref{prop:rho}
imply that 
\begin{equation*}
\frac{T_i D_F(n)}{D_F(n)} = -\chi_i(\lambda_i^-;n) = \rho_i(n)
\frac{G_i(\lambda_i^+;n)}{G_i(\lambda_i^-;n)}.
\end{equation*}
Due to equation \eqref{eq:rho-tau} we have
\begin{equation*}
\frac{T_i \tau(n)}{\tau(n)} = \frac{T_i D_F(n)}{D_F(n)} \prod_{j=1,j\ne i}^N \left(
\frac{L_i(\lambda_j^-)}{L_i(\lambda_j^+)} \right)^{n_j},
\end{equation*}
which upon integration gives the statement of the theorem.
\end{proof}

\section{Conclusion and remarks}
We have shown that within the $\bar\partial$-dressing method
the $\tau$-function of the quadrilateral lattice can be identified with the 
Fredholm determinant of the integral equation
inverting the corresponding 
nonlocal $\bar\partial$ problem. The proof of this fact was quite elementary, 
and in a certain aspect (in the continuous limit $\lambda_i^-$ and 
$\lambda_i^+$ coincide, which produces essental singularities in the 
function $G(\lambda)$) even simpler then in the earlier similar paper 
\cite{tau-cn} where the $\tau$-function of the conjugate
nets had been treated. 

It should be noted that the prescribed singularity structure of the functions
$\psi_i(\lambda;n)$ is the same like in the construction of the
algebro-geometric solutions of the discrete Darboux equations \cite{AD-Chicago},
and the $\bar\partial$-dressing method meaning of the first potentials $\rho_i$
given here is a direct counterpart of that given in \cite{AD-Chicago}. 

Finally, we remark that geometrically one can consider quadrilateral lattices in
projective spaces over division rings \cite{gaql} and a substantial part of the
integrability structures (including existence of the first potentials $\rho_i$ 
or the Darboux-type transformations) goes over.
However, it seems that without certain additional (commutativity) conditions 
one cannot define a $\tau$-function of such quadrilateral lattice over general
division ring.  

%%%%%%%%%%%%%%%%%%%%%%%%%%%%%%%%%%%%%%%%%%%%%%%%%%%%%%%%%%%%%%%%%%%%%%%%%%
\bibliographystyle{amsplain}

\begin{thebibliography}{10}


\bibitem{AYF}
M. J. Ablowitz, D. Bar Yaacov, A. S. and Fokas, \emph{On the inverse scattering
problem for the {Kadomtsev}--{Petviashvili} equation}, Stud. Appl.
Math. {\bf 69} (1983), 135--143.

\bibitem{ASvM}
M. Adler, T. Shiota, and P. van Moerbeke, \emph{Random matrices, Virasoro
algebras and noncommutative KP}, Duke Math. J. \textbf{94}, (1998), 379--431.

\bibitem{BogdanovManakov}
L. V. Bogdanov,  S. V. Manakov, {\it The nonlocal $\bar{\partial}$-problem
and (2+1)-dimensional soliton equations}, J. Phys. A: Math. Gen. {\bf 21}
(1988), L537--L544.


\bibitem{BoKo}
L. V. Bogdanov, and  B. G. Konopelchenko, 
{\it Lattice and q-difference Darboux--Zakharov--Manakov systems via
$\bar{\partial}$ method},
J. Phys. A: Math. Gen. {\bf 28} (1995), L173-L178.

\bibitem{BK-mKP}
L. V. Bogdanov, and B. G. Konopelchenko, \emph{Analytic-bilinear approach to 
integrable hierarchies {II}.  {Multicomponent} {KP} and {2D} {Toda} 
lattice hierarchies}, J. Math. Phys. {\bf 39} (1998), 4701--4728.

\bibitem{BorodinDeift}
A. Borodin, and P. Deift, \emph{Fredholm determinants, Jimbo--Miwa--Ueno
$\tau$-functions, and representation theory}, Comm. Pure Appl. Math. \textbf{55}
(2002), 1160--1230.

\bibitem{DarbouxOS}
G. Darboux, \emph{Le\c{c}ons sur les syst\'{e}mes orthogonaux et les
coordonn\'{e}es curvilignes}, Gauthier-Villars, Paris, 1910.


\bibitem{DKJM}
E. Date, M. Kashiwara, M. Jimbo, and T. Miwa, \emph{Transformation groups for
soliton equations}, [in:] Nonlinear integrable systems --- classical theory and
quantum theory, Proc. of RIMS Symposium, M. Jimbo and T. Miwa (eds.), World
Scientific, Singapore, 1983, 39--119.

\bibitem{DCN}
A. Doliwa, {\it Geometric discretisation of the Toda system},
Phys. Lett. A {\bf 234} (1997), 187--192.

\bibitem{AD-Chicago}
A. Doliwa, {\it Integrable multidimensional
discrete geometry:  
quadrilateral lattices, their transformations and reductions}, 
[in:]
{\it Integrable hierarchies and modern physical theories} H. Aratyn \& A. S.
Sorin (eds.), Kluwer, Dordrecht, 2001 pp. 355--389.

\bibitem{tau-cn}
A. Doliwa, {\it On $\tau$--function of conjugate nets}, 
J. Nonlin. Math. Phys. {\bf 12} Supplement (2005)
244--252.

\bibitem{gaql}
A. Doliwa, \emph{Geometric algebra and quadrilateral lattices},
\texttt{arXiv: 0801.0512 [nlin.SI]}. 

\bibitem{MQL}
A.~Doliwa and P.~M. Santini, \emph{Multidimensional quadrilateral lattices 
are integrable}, Phys. Lett. A \textbf{233} (1997), 365--372.

\bibitem{DS-sym}
A.~Doliwa and P.~M. Santini, \emph{The symmetric, {D}-invariant and {E}gorov 
reductions of the
  quadrilateral lattice}, J. Geom. Phys. \textbf{36} (2000), 60--102.


\bibitem{DMS}
A.~Doliwa, S.~V. Manakov and P.~M. Santini, 
\emph{$\bar\partial$-reductions of
  the multidimensional quadrilateral lattice: the multidimensional 
  circular
  lattice}, Comm. Math. Phys. \textbf{196} (1998), 1--18.

\bibitem{TQL}
A.~Doliwa, P.~M.~Santini and M.~Ma\~nas,
\emph{Transformations of quadrilateral lattices}, J. Math. Phys. \textbf{ 41} 
(2000), 944--990. 


\bibitem{DMMMS}
A. Doliwa, M. Ma{\~n}as, L. Mart{\'\i}nez Alonso, E. Medina, and P. M.
Santini, \emph{Multicomponent {KP} hierarchy and classical transformations of
conjugate nets}, J. Phys. A {\bf 32} (1999), 1197--1216.

\bibitem{GrinevichOrlov}
P. G. Grinevich, and A. Yu. Orlov, \emph{Flag spaces in {KP} theory and
{V}irasoro action on $\det \bar\partial_j$ and {S}egal--{W}ilson
$\tau$-function}, [in:] Problems of Modern Quantum Field Theory, 
Springer-Verlag, Berlin, 1989, 86--106.


\bibitem{Hirota-book}
R. Hirota, \emph{The direct method in soliton theory}, Cambridge University
Press, 2004.

\bibitem{KvLJMP}
V. G. Kac, and J.~van~de Leur, \emph{The n-component {KP} hierarchy and
representation theory}, J. Math. Phys. {\bf 44} (2003), 3245--3293.

\bibitem{Kharchev}
S. Kharchev, \emph{Kadomtsev--Petviashvili hierarchy and generalized Kontsevich
model}, Amer. Math. Soc. Transl. Ser. 2, \textbf{191}, AMS, rovidence, 1999.

\bibitem{Konopelchenko}
B. G. Konopelchenko, {\it Solitons in multidimensions. Inverse spectral
transform method}, World Scientific,
Singapore, 1993.

\bibitem{KoSchief2}
B.~G. Konopelchenko and W.~K. Schief, \emph{Three-dimensional integrable
  lattices in {Euclidean} spaces: Conjugacy and orthogonality}, Proc. Roy. Soc.
  London A \textbf{454} (1998), 3075--3104.

\bibitem{MMAM-cn}
M. Ma{\~n}as, L. Mart{\'\i}nez Alonso, and E. Medina, 
\emph{Dressing methods for geometric
nets:~I. Conjugate nets}, J. Phys. A: Math. Gen. {\bf 33}
(2000), 2871--2894.

\bibitem{Moerbeke}
P. van Moerbeke, \emph{Integrable Lattices: Random matrices and Random 
Permutations}, [in:]
Random Matrices and their Applications, MSRI Publications {\bf 40} (2001),
321-406. 

\bibitem{Palmer}
J. Palmer, \emph{Determinants of Cauchy--Riemann operators as $\tau$-functions},
Acta Appl. Math. \textbf{18} (1990), 199-223.

\bibitem{PoppeSattinger}
Ch. P\"{o}ppe, and D. H. Sattinger, \emph{Fredholm determinants and the $\tau$
function for the Kadomtsev--Petviashvili hierarchy}, Publ. RIMS, Kyoto
Univ. {\bf 24} (1988), 505--538.


\bibitem{Sauer2}
R.~Sauer, \emph{Projective Liniengeometrie}, de Gruyter, Berlin--Leipzig, 1937.

\bibitem{Sauer}
R.~Sauer, \emph{Differenzengeometrie}, Springer, Berlin, 1970.

\bibitem{Sato}
M. Sato, \emph{Soliton equations as dynamical systems on infinite dimensional
{G}rassman manifolds}, RIMS Kokyuroku {\bf 439} (1981), 30--46.

\bibitem{SegalWilson}
G. Segal, and G. Wilson, 
\emph{Loop groups and equations of {K}d{V} type}, Inst.
Hautes \'Etudes Sci. Publ. Math. {\bf 61} (1985), 5--65.

\bibitem{Smithies}
F. Smithies, \emph{Integral Equations}, Cambridge Univ. Press, Cambridge, 1965.

\bibitem{ZakMa}
V.~E. Zakharov and S.~V. Manakov, \emph{Construction of higher--dimensional
  nonlinear integrable systems and of their solutions}, Funk. Anal. Appl.
  \textbf{19} (1985), 89--101.

\bibitem{ZakMa2}
V.~E. Zakharov and S.~V. Manakov, \emph{Reductions in systems integrated 
by the method of the inverse
  scattering problem}, Dokl. Math. \textbf{57} (1998), 471--474.

\bibitem{ZMNP}
V. E. Zakharov, S. V.Manakov, S. P. Novikov, and L. P. Pitaevskii, 
\emph{Theory of
solitons --- inverse scattering metod}, Nauka, Moscow, 1980.



\end{thebibliography}

\end{document}